\documentclass[letterpaper, 10 pt, conference]{ieeeconf}  

\IEEEoverridecommandlockouts                              
\usepackage{graphics} 
\usepackage{svg}
\usepackage{epsfig} 
\usepackage{times} 
\usepackage{amsmath} 
\usepackage{amssymb}  

\usepackage{amsthm}
\usepackage{subfiles}
\usepackage{cite}
\renewcommand{\epsilon}{\varepsilon}

\newtheorem{theorem}{Theorem}

\newtheorem{corollary}{Corollary}
\newtheorem{lemma}{Lemma}

\newtheorem{definition}{Definition}

\newtheorem{remark}{Remark}

\newtheorem{example}{Example}

\renewcommand{\epsilon}{\varepsilon}
\renewcommand{\phi}{\varphi}

\newcommand{\lf}{\left}
\newcommand{\rg}{\right}

\newcommand{\ii}{\infty}

\newcommand{\pp}[1]{\frac{\partial}{\partial#1}}

\newcommand{\dee}{\,\mathrm{d}}

\newcommand{\sbs}{\subseteq}

\newcommand{\ssbs}{\subset}

\newcommand{\tms}{\times}

\newcommand{\ba}{\bar}

\newcommand{\N}{\mathbb{N}}

\newcommand{\R}{\mathbb{R}}

\newcommand{\Nc}{\mathcal{N}}

\newcommand{\al}{\alpha}
\newcommand{\bt}{\beta}

\newcommand{\te}{\theta}

\newcommand{\kp}{\kappa}
\newcommand{\lm}{\lambda}

\newcommand{\rh}{\rho}

\newcommand{\ta}{\tau}

\newcommand{\ph}{\varphi}

\newcommand{\ps}{\psi}

\newcommand{\Om}{\Omega}

\newcommand{\textall}{\,\,\text{for all }}

\newcommand{\K}{\mathcal{K}}
\newcommand{\KL}{\mathcal{K}\mathcal{L}}

\newcommand{\cbf}{h}
\newcommand{\vcbf}{h}
\newcommand{\cbvf}{v}
\newcommand{\barrier}{B}
\newcommand{\terminalpayoff}{g}

\newcommand{\Rn}{\R^n}
\newcommand{\Rm}{\R^m}

\newcommand{\Rge}{\R_{\ge 0}}

\newcommand{\xsig}{\mathbf{x}}
\newcommand{\usig}{\mathbf{u}}

\newcommand{\uvals}{\mathcal{U}}
\newcommand{\usigs}{\mathbb{U}}

\newcommand{\safeset}{\mathcal{S}}

\newcommand{\traj}{\xsig_x^\usig}
\newcommand{\tildetraj}{\xsig_x^{\tilde{\usig}}}
\newcommand{\startraj}{\xsig_x^{\usig^*}}

\newcommand{\pdesol}{\mu}
\newcommand{\hamiltonian}{H_\al}
\newcommand{\vanillahamiltonian}{H}

\title{\LARGE \bf
Viscosity CBFs: Bridging the Control Barrier Function and Hamilton-Jacobi Reachability Frameworks in Safe Control Theory
}

\author{Dylan Hirsch, Jaime Fern\'{a}ndez Fisac, and Sylvia Herbert
\thanks{Research reported in this publication was supported by NIBIB of the National Institutes of Health under award number T32EB009380. The content is solely the responsibility of the authors and does not necessarily represent the official views of the National Institutes of Health.
This research was 100\% financed by Federal Grants.} %
\thanks{Dylan Hirsch (corresponding author) and Sylvia Herbert are with the Department of Mechanical and Aerospace Engineering, University of California at San Diego, 9500 Gilman Drive, La Jolla, CA 92093.
        {\tt\small dhirsch@ucsd.edu, sherbert@ucsd.edu.}}%
\thanks{Jaime Fisac is with the Department of Mechanical and Aerospace Engineering, Princeton University, Princeton, NJ 08544. {\tt\small jfisac@princeton.edu.}}%
}

\begin{document} 

\maketitle
\thispagestyle{empty}
\pagestyle{empty}

\begin{abstract}
Control barrier functions (CBFs) and Hamilton-Jacobi reachability (HJR) are central frameworks in safe control.
Traditionally, these frameworks have been viewed as distinct, with the former focusing on optimally safe controller design and the latter providing sufficient conditions for safety.
A previous work introduced the notion of a control barrier value function (CB-VF), which is defined similarly to the other value functions studied in HJR but has certain CBF-like properties.
In this work, we proceed the other direction by generalizing CBFs to non-differentiable ``viscosity'' CBFs.
We show the deep connection between viscosity CBFs and CB-VFs, bridging the CBF and HJR frameworks.
Through this bridge, we characterize the viscosity CBFs as precisely those functions which provide CBF-like safety guarantees (control invariance and smooth approach to the boundary).
We then further show nice theoretical properties of viscosity CBFs, including their desirable closure under maximum and limit operations.
In the process, we also extend CB-VFs to non-exponential anti-discounting and update the corresponding theory for CB-VFs along these lines.
\end{abstract}

\section{Introduction}

Control barrier functions (CBFs) and Hamilton-Jacobi reachability (HJR) are two primary theoretical frameworks for safe control \cite{Ames-Tabuada-TAC-CBFs-QP-2017,Ames-Tabuada-ECC-CBFs-Theory-and-Applications-2019,Mitchell-Tomlin-TAC-HJR-2005,Bansal-Tomlin-CDC-HJR-Overview-2017}.
The central notion in HJR is the value function, a scalar function that describes the ability of an optimal controller to achieve a task.
The payoff functionals used in HJR (e.g. maximum-over-time, minimum-over-time) correspond to various types of tasks (e.g. target-reaching, obstacle-avoidance) \cite{Mitchell-Tomlin-TAC-HJR-2005,Margellos-Lygeros-TAC-Reach-Avoid-2011,Fisac-Sastry-Reach-Avoid-2015}.
HJR is \textit{constructive} in that once the payoff is chosen, the value function is uniquely specified.
The value functions considered in HJR are usually time-dependent (although time-independent alternatives also exist \cite{Altarovici-Zidani-ESAIM-State-Constrained-2012}), are typically non-differentiable, and are characterized by a certain weak (``viscosity'') solution to a Hamilton-Jacobi Partial Differential Equation (HJ-PDE).

CBFs are also scalar functions that encode safety information, but they are defined based on the satisfaction of the ``CBF inequality,'' a particular partial differential inequality. 
CBFs provide sufficient conditions to certify not only that a certain set is control-invariant, but that one can avoid approaching the boundary of this set too quickly.
Like Lyapunov functions, CBFs are \textit{non-constructive} in the sense that there may be many CBFs for a system and there is no single standard approach to obtain a particular one which provides desired safety guarantees.
Moreover, CBFs are usually time-independent (although time-varying extensions exist \cite{Xiao-Christos-TAC-Adaptive-CBFs-2022}) and continuously differentiable (although notions of non-differentiable CBFs also exist \cite{Glotfelter-Egerstedt-LCSS-NBFs-2017}, \cite{Glotfelter_2021}).

In \cite{Choi-Herbert-CDC-CB-VFs-2021}, the authors introduce control barrier value functions (CB-VFs) by anti-discounting the payoff for an obstacle-avoidance task.
The CB-VF is a value function in the sense of HJR, but it has several CBF-like properties.

In this work, we proceed in the other direction by introducing the notion of a viscosity CBF, a generalization of a CBF that need not be differentiable (and indeed generalizes the non-differentiable CBFs in \cite{Glotfelter-Egerstedt-LCSS-NBFs-2017}, \cite{Glotfelter_2021}).
Unlike CB-VFs, viscosity CBFs are defined in the spirit of traditional CBFs, i.e. via a solution to a CBF inequality. Whereas traditional CBFs must satisfy this inequality in a classical sense, we only require viscosity CBFs to do so in a certain weak (``viscosity'') sense.

\begin{figure}[!t]
    \centering
    \includegraphics[width=\columnwidth]{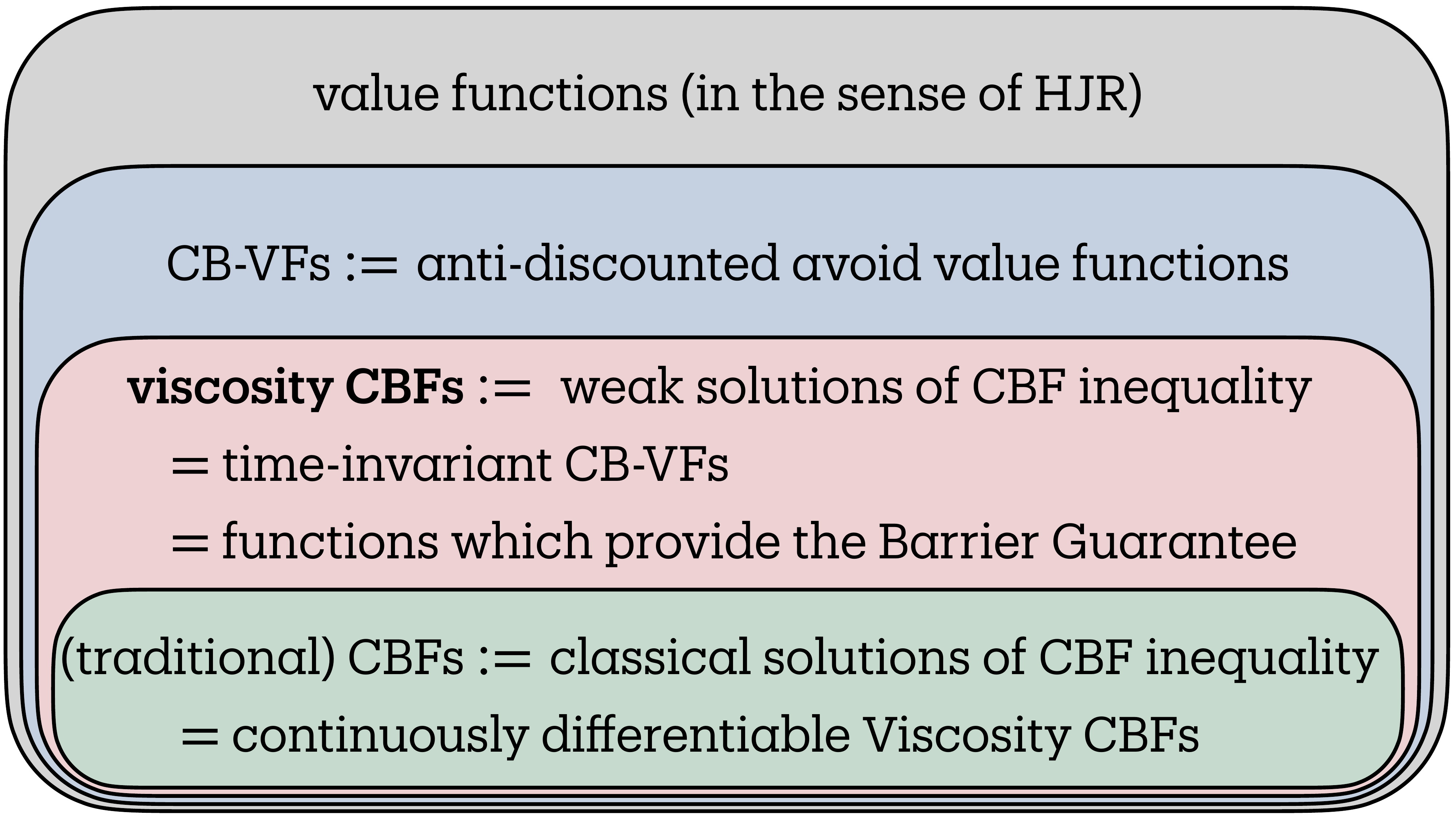}
    \caption{\textbf{Graphical abstract:} in this work, we bridge Hamilton-Jacobi reachability (HJR) and control barrier function (CBFs).
    To do so, we extend the work on control barrier value functions (CB-VFs) in \cite{Choi-Herbert-CDC-CB-VFs-2021} and also introduce the notion of a viscosity CBF.
    A CB-VF is defined similarly to an avoid value function (a value function used in HJR for obstacle-avoidance tasks), except that it is anti-discounted according to a class $\mathcal{K}$ function $\alpha$.
    A viscosity CBF is similar to a standard CBF, except that the usual CBF inequality is only required to hold in a weak (``viscosity'') sense.
    We show that a viscosity CBF is equivalent to a time-invariant CB-VF.
    We then show the set of viscosity CBFs is precisely the set of continuous functions which provide the Barrier Guarantee, i.e. certify control invariance of a set and bounds on the speed at which the system can approach the set boundary.
    }
    \label{fig:graphical-abstract}
\end{figure}

We show that the set of viscosity CBFs for a system are in fact precisely the set of time-invariant CB-VFs for the system, theoretically bridging the CBF and HJR frameworks.
Through this bridge, we demonstrate that we can characterize the set of all continuous functions (including those which are non-differentiable and even not locally Lipschitz) that can provide CBF-like guarantees (Fig. \ref{fig:graphical-abstract}).

We then proceed to further study these viscosity CBFs, including describing how one can synthesize new viscosity CBFs by taking limits of known viscosity CBFs or traditional CBFs.
In the process, we also extend the notion of a CB-VF beyond exponential anti-discounting to anti-discounting determined by nonlinear class $\mathcal{K}$ functions.

We note that in \cite{Camilli_2008}, the authors take a viscosity-based approach to control Lyapunov functions.
Although related, this work differs mathematically and conceptually from the present one.
In particular, the authors use a positive definite function of the state to encourage stabilization to a point, whereas our work uses a class $\K$ function of the value to prevent rapid approach of the boundary, fundamentally changing the underlying HJ-PDE.
Additionally, the authors in \cite{Camilli_2008} analyze an optimal control problem with a running-cost payoff, whereas we analyze one with a CB-VF payoff to make the connection between HJR and CBFs.

\section{Problem Setting and Preliminaries}

In this work we consider a dynamical system
\begin{equation}\label{eqn:dynamics}
    \dot{\xsig} = f(\xsig,\usig),
\end{equation}
where $f:\Rn \tms \uvals \to \Rn$ is Lipschitz, with the set of allowed control input values $\uvals \ssbs \Rn$ assumed compact.
Here, $\xsig$ is the system's state and $\usig$ is the control input (we use the boldface to distinguish the state trajectory $\xsig$ and input signal $\usig$ from a state value $x \in \Rn$ and input value $u \in \uvals$).

We will denote the set of measurable functions $\usig: \Rge \to \uvals$ by $\usigs$, representing the allowed control signals.

\begin{definition}[Class $\K$; Definition 4.2 in Chapter 4.4 of \cite{Khalil-Nonlinear-Systems}] A function $\al: \Rge \to \Rge$ is \textup{class} $\K$ if it is continuous, is strictly increasing, and satisfies $\al(0) = 0$.
\end{definition}

\begin{definition}[Class $\KL$; Definition 4.3 in Chapter 4.4 of \cite{Khalil-Nonlinear-Systems}]
    A function $\bt:\Rge \tms \Rge \to \Rge$ is class $\KL$ if it is continuous and (i) $\bt(\cdot,t)$ is class $\K$ for each $t \ge 0$, (ii) $\bt(r,\cdot)$ is non-increasing for each $r \ge 0$, and (iii) $\bt(r,t) \to 0$ as $t \to \ii$ for each $r \ge 0$.
\end{definition}

By Lemma 4.4 of \cite{Khalil-Nonlinear-Systems}, if a class $\K$ function $\al$ is locally Lipschitz, we can associate it with a class $\KL$ function $\bt_\al$ by for each $r \ge 0$ defining $\bt_\al(r,\cdot)$ to be the solution of the initial value problem $\dot{y} = -\al(y)$, $y(0) = r$.

\begin{definition}[Control Invariant Set]
    A set $\safeset \sbs \Rn$ is control invariant (with respect to $f$) if for each $x \in \safeset$ there is a $\usig \in \usigs$ such that $\traj(t) \in \safeset$ for all $t \ge 0$.
\end{definition}

\begin{definition}[Control Barrier Function]\label{def:control-barrier-function}
    A function $\cbf: \Rn \to \R$ is a \textup{control barrier function} with respect to a class $\K$ function $\al$ if it is continuously differentiable and
    \begin{equation}\label{eqn:cbf-inequality}\tag{CBF inequality}
        \max_{u \in \uvals} \nabla \cbf(x) \cdot f(x,u) \ge - \al(\cbf(x))
    \end{equation}
    for each $x \in \Rn$ satisfying $\cbf(x) > 0$.
\end{definition}

\begin{remark}
    Some definitions of CBFs do not impose demands on the behavior of the system outside of the control-invariant set, while others require that if the system begins outside of this set it can return sufficiently quickly toward the boundary.
    For ease of exposition, we adopt the former definition (hence we only require that the CBF inequality holds when $\cbf(x) > 0$).
    We note, however, that analogous results can also be shown in a similar fashion for the latter definition (which will be saved for future work).
\end{remark}

\section{The Barrier Guarantee}
While CBFs can be used to certify control-invariance of a set, if one is solely concerned with this objective, Nagumo's theorem can be used instead.
The true upsell of CBFs is their additional guarantees, in particular bounding the rate at which a system can approach the boundary of a safe set (in some respects, it is this property that enables safety filter design using CBFs).
Motivated by this idea, we introduce the notion of a Barrier Guarantee (c.f. Definition 4 in \cite{Glotfelter-Egerstedt-LCSS-NBFs-2017}).

\begin{definition}[Barrier Guarantee]\label{def:barrier-guarantee}
    A continuous function $\barrier:\Rn \to \R$ provides the \textup{Barrier Guarantee} with respect to a locally Lipschitz class $\K$ function $\al$ if for each $\te \in [0,1)$ and each $x \in \Rn$ satisfying $\barrier(x) > 0$, there is a $\usig \in \usigs$ such that for all $t \ge 0$ we have
    \begin{equation}\label{eqn:barrier-guarantee}
        \barrier \lf( \traj(t)  \rg) \ge \bt_\al(\te \barrier(x), t).
    \end{equation}
\end{definition}

\begin{remark}\label{remark:barrier-guarantee-motivation}
    To motivate the preceding definition, suppose as in Theorem 2 of \cite{Ames-Tabuada-ECC-CBFs-Theory-and-Applications-2019}, we have a (traditional) CBF $h:\Rn \to \R$ and that we can find a Lipschitz feedback law $u = k(x)$ such that for each $x \in \Rn$ we have
    $$\nabla h(x) \cdot f(x, k(x)) \ge -\al\lf(h(x)\rg).$$
    Let $\xsig$ be some trajectory of the closed loop system.
    Recognizing that the left-hand-side of the above inequality is simply $\frac{d}{dt}h(\xsig(t))$, the comparison principle then gives
    $$h\lf(\xsig(t)\rg) \ge \bt_\al(\cbf(x), t)$$
    for all $t \ge 0$, where $x := \xsig(0)$.
    The above inequality is almost identical to \eqref{eqn:barrier-guarantee} with $\barrier = \cbf$, except that in \eqref{eqn:barrier-guarantee} the multiplier $\te$ appears for technical reasons related to the fact that a safe Lipschitz controller for a CBF may not exist.
    
    For an example when a safe Lipschitz controller does not exist, consider the CBF $\cbf(x) = 1 - x^2$ for the scalar system $\dot{\xsig} = \usig$ with the admissible control set $\uvals = \{-1,+1\}$.
    There are only two Lipschitz controllers: $k(x) \equiv +1$ and $k(x) \equiv -1$, neither of which is safe.
\end{remark}

If a function $\barrier$ provides the Barrier Guarantee, it immediately implies that the strict zero super-level set $\safeset:= \{x \mid \barrier(x) > 0\}$ of $\barrier$ is control invariant.
(Note that we consider throughout this work the \textit{strict} zero super-level because the existence of even a traditional CBF does not in fact guarantee control-invariance of the \textit{non-strict} zero super-level set without additional assumptions, e.g. control-affine dynamics and convexity of $\uvals$; see Appendix \ref{subsection:appendix-counter-example} for a counter-example.)
Beyond just control invariance, however, as demonstrated in Remark \ref{remark:barrier-guarantee-motivation}, the Barrier Guarantee provides a CBF-like bound on how fast the system can approach the boundary $\partial \safeset$ of this set whenever it is initialized within $\safeset$.

While any CBF provides the Barrier Guarantee, the continuous differentiability requirements of traditional CBFs are stronger than needed to guarantee this property holds.
Our primary goal in this work will be to characterize all those functions which provide this guarantee.
To do so, we will introduce a certain notion of a non-differentiable CBF, slightly more general than that introduced in \cite{Glotfelter-Egerstedt-LCSS-NBFs-2017} and having several nice theoretical properties.
It will turn out that it is precisely these more general ``viscosity'' CBFs which provide the Barrier Guarantee.
To establish this fact, we will need to elucidate the deep connection between CBFs and the value functions of HJR.
In some sense, characterizing those functions which provide the Barrier Guarantee then serves as an example of the utility of bridging CBF and HJR theory.

The following easily verified fact will be of later use.
\begin{lemma}\label{lemma:non-negative-barrier-guarantee}
    Let $\barrier:\Rn \to \R$ be continuous and $\al$ be locally Lipschitz class $\K$.
    Then $\barrier$ provides the Barrier Guarantee w.r.t. $\al$ if and only if $\max\{0,\barrier\}$ provides the Barrier Guarantee w.r.t. $\al$.
\end{lemma}

\section{Viscosity CBFs and Control Barrier Value Functions}

\subsection{Viscosity CBFs}

To characterize the functions which provide the Barrier Guarantee and also to connect CBFs and HJR, we introduce the key mathematical object of this work: the viscosity CBF.

\begin{definition}[Viscosity CBF] \label{def:viscosity-cbfs}
    A continuous function $\vcbf: \Rn \to \R$ is a \textit{viscosity CBF} with respect to a class $\K$ function $\al$ if it is continuous and for each $x \in \Rn$, the inequality
    \begin{equation}\label{eqn:viscosity-cbf-1}
        \max_{u \in \uvals} \nabla \vcbf(x) \cdot f(x,u) \ge - \al(\vcbf(x))
    \end{equation}
    holds in the viscosity sense at each $x \in \Rn$ for which $\vcbf(x) > 0$.
    By ``in the viscosity sense,'' we mean that for any continuously differentiable function $\ph:\Rn \to \R$ such that $\vcbf - \ph$ has a local maximum at $x$, we have
    \begin{equation}\label{eqn:viscosity-cbf-2}
        \max_{u \in \uvals} \nabla \ph(x) \cdot f(x,u) \ge - \al(\vcbf(x)).
    \end{equation}
    (``In the viscosity sense'' is standard terminology in PDE literature; see e.g. Chapter II.5.5 in \cite{Bardi-Dolcetta-Optimal-Control}.)
\end{definition}

\begin{example}
    Consider the scalar system with $\dot{\xsig} = \xsig + \frac{\xsig+\xsig^3}{1 + |\xsig|} \usig$, with admissible control set $\uvals = [-1,+1]$.
    Let $\vcbf:\R \to \R$ be the signed distance function to an unsafe set $(-\ii,-1] \cup [+1,+\ii)$, i.e. $\vcbf(x) = 1 - |x|$.
    
    This function is not a CBF as it is not differentiable at $0$, but we show it is a viscosity CBF with respect to $\al(r) = r$.
    One can check directly that \eqref{eqn:viscosity-cbf-1} holds in a classical sense at each $x \in (-1,0) \cup (0,+1)$, and thus it also holds in a viscosity sense here (Lemmas II.1.7(a) and II.1.8(b) in \cite{Bardi-Dolcetta-Optimal-Control}).

    At the point $x = 0$ on the other hand, it is immediate that \eqref{eqn:viscosity-cbf-2} holds for every continuously differentiable $\ph:\R \to \R$, and thus it certainly holds for such $\ph$ for which $\vcbf - \ph$ has a local maximum at $0$.
    
\end{example}

Note that if a viscosity CBF happens to be continuously differentiable, then it is also a traditional CBF (see Lemmas II.1.7(a) and II.1.8(b) in \cite{Bardi-Dolcetta-Optimal-Control}).
Similarly, all traditional CBFs are viscosity CBFs (see Proposition II.1.3(b) in \cite{Bardi-Dolcetta-Optimal-Control}).
It will often be most convenient to work with non-negative viscosity CBFs.
Doing so is made simple thanks to the following lemma, whose proof is trivial.
\begin{lemma}\label{lemma:non-negative-viscosity-cbf}
    Let $\vcbf:\Rn \to \R$ be continuous and $\al$ be class $\K$.
    Then $\vcbf$ is a viscosity CBF w.r.t. $\al$ if and only if $\max\{0,\vcbf\}$ is a viscosity CBF w.r.t. $\al$.
\end{lemma}

Viscosity CBFs are a particular way of generalizing traditional CBFs to non-differentiable functions.
They are slightly more general than the non-smooth CBFs introduced in \cite{Glotfelter-Egerstedt-LCSS-NBFs-2017}, permitting even CBFs that are not locally Lipschitz.
The purpose of allowing for such exotic non-differentiable CBFs is not out of mathematical interest but for their useful theoretical properties and direct connection to HJR and the Barrier Guarantee.
To establish these properties and connections, we must introduce a certain kind of value function from HJR.

\subsection{Viscosity solutions of Hamilton Jacobi Partial Differential Equations}

Before proceeding, we briefly review the notion of viscosity solutions.
A viscosity solution is a specific type of weak solution of a HJ-PDE which need not be continuously differentiable (often too restrictive for existence) but satisfies more than just differentiability almost everywhere (often too loose for uniqueness).
\begin{definition}[Viscosity solution]
Let $\Om \sbs \Rm$ be open, let $M \sbs \R$ be closed, and let $F:\Om \tms M \tms \Rm \to \R$ be continuous.
Consider the HJ-PDE
\begin{equation*}
    F\lf(z,\pdesol, D \pdesol\rg) = 0.
\end{equation*}
A continuous function $\pdesol:\Om \to M$ is a 
\begin{itemize}
\item \textup{viscosity sub-solution} of this PDE if for each continuously differentiable function $\ph:\Om \to \R$, we have 
$F(z_0,\pdesol(z_0), D \ph(z_0)) \le 0$ for any $z_0 \in \Om$ at which $\pdesol - \ph$ has a local maximum.

\item \textup{viscosity super-solution} of this PDE if for each continuously differentiable function $\ph:\Om \to \R$, we have 
$F(z_0,\pdesol(z_0), D \ph(z_0)) \ge 0$ for any $z_0 \in \Om$ at which $\pdesol - \ph$ has a local minimum.

\item \textup{viscosity solution} of this PDE if it is both a viscosity sub-solution and super-solution.
\end{itemize}

\end{definition}
Notationally, for the time-dependent PDEs considered in this work, it will be the case that $\pdesol = \cbvf$, $z = (x,T)$, $m = n+1$, $\Om = \Rn \tms (0,\ii)$, $M = \Rge$, and $D = (\nabla,\pp{T})$.
For the time-independent PDEs considered, it will be the case that $\pdesol = \vcbf$, $z = x$, $m = n$, $\Om = \Rn$, $M = \Rge$, and $D = \nabla$.
For a more thorough introduction to viscosity solutions of HJ-PDEs, we refer the reader to \cite{Evans-PDEs,Bardi-Dolcetta-Optimal-Control,Fleming-Soner-Viscosity}.

When we work in particular with non-negative viscosity CBFs, we get the following useful characterization.
\begin{lemma}\label{lemma:viscosity-cbf-is-hj-pde-sub-solution}
    Let $\vcbf:\Rn \to \Rge$ be continuous and $\al$ be class $\K$.
    Define the Hamiltonian $\hamiltonian:\Rn \tms \Rge \tms \Rn \to \R$ by 
    \begin{equation}\label{eqn:hamiltonian}
        \hamiltonian(x,r,\lm) = \max_{u \in \uvals} \lm \cdot f(x,u) + \al(r).
    \end{equation}
    Then $\vcbf$ is a viscosity CBF w.r.t. $\al$ if and only if it is a viscosity sub-solution of $-\hamiltonian(x,\vcbf,\nabla \vcbf) = 0$.
\end{lemma}
\begin{remark}
In the above HJ-PDE, the negative signs must be retained as negating both sides of an HJ-PDE can result in unwanted changes to its viscosity sub/super-solutions.
\end{remark}
\begin{proof}
    The reverse implication is immediate, so suppose that $\vcbf$ is a viscosity CBF w.r.t. $\al$.
    Fix $x_0 \in \Rn$, and let $\ph:\Rn \to \R$ be a continuously differentiable function such that $\vcbf - \ph$ has a local maximum at $x_0$.
    We wish to show that
    \begin{equation}\label{proof:viscosity-cbf-subsolution-wwts-1}
        -\max_{u \in \uvals} \nabla \ph(x_0) \cdot f(x_0,u) - \al(\vcbf(x_0)) \le 0.
    \end{equation}
    There are two cases: either $\vcbf(x_0) > 0$ or $\vcbf(x_0) = 0$.
    In the first case, \eqref{proof:viscosity-cbf-subsolution-wwts-1} follows directly from the definition of a viscosity CBF.
    In the second case, since $\vcbf$ is non-negative, it follows that $\vcbf$ has a minimum at $x_0$.
    But since $\ph - \vcbf$ has a local minumum at $x_0$, it follows that $\ph$ has a local minimum at $x_0$ as well.
    Thus $\nabla \ph(x_0) = 0$, so \eqref{proof:viscosity-cbf-subsolution-wwts-1} follows.
\end{proof}

\subsection{Avoid value functions and control barrier value functions}

In HJR, the standard value function used for obstacle avoidance tasks is the ``avoid'' value function.
Given a continuous function $\terminalpayoff:\Rn \to \R$, the avoid function $V:\Rn \tms \Rge \to \R$ corresponding to $\terminalpayoff$ is defined by
\begin{equation*}
    V(x,T) = \sup_{\usig \in \usigs} \min_{t \in [0,T]} g(\traj(t)).
\end{equation*}
In applications, we choose $g$ to be non-positive within and only within the obstacle so that $V(x,T) > 0$ if and only if the system, initialized at $x$, can avoid the obstacle until time $T$.
We refer to such a $g$ as an ``immediate safety function.''

In \cite{Choi-Herbert-CDC-CB-VFs-2021}, the authors define a time-varying value function having some CBF-like properties, which they call a control barrier value function (CB-VF).
A CB-VF is similar to an avoid value function except that it is anti-discounted according to a function $\al$.
As the previous work only considers exponential anti-discounting, to fully capture the striking connection between CBFs and HJR, we must extend the CB-VF definition to more general anti-discounting, i.e. as specified by a locally Lipschitz class $\K$ function $\al$ (where exponential anti-discounting corresponds to $\al(r) = \gamma r$).

Given $\terminalpayoff:\Rn \to \Rge$ continuous and $\al$ locally Lipschitz class $\K$,
we define the CB-VF $\cbvf:\Rn \tms \Rge \to \Rge$ corresponding to $\vcbf$ and $\al$ implicitly by
\begin{equation}\label{eqn:cbvf}
    \bt_\al\lf(\cbvf(x,T),T\rg) = \sup_{\usig \in \usigs} \min_{t \in [0,T]} \bt_\al \lf( \terminalpayoff\lf(\traj(t)\rg), T-t \rg).
\end{equation}
Note that for each $x \in \Rn$ and $T \ge 0$, \eqref{eqn:cbvf} indeed has a unique solution $\cbvf(x,T)$ since 
$$0 \le \sup_{\usig \in \usigs} \min_{t \in [0,T]} \bt_\al \lf( \terminalpayoff\lf(\traj(t)\rg), T-t \rg) \le \bt_\al(\terminalpayoff(x),T).$$

It can be seen from \eqref{eqn:cbvf} that if there is a control signal $\usig$ for which $\terminalpayoff\lf(\traj(t)\rg) \ge \bt_\al(\terminalpayoff(x),t)$ for all $t \in [0,T]$, then $\cbvf(x,T) = \terminalpayoff(x)$.
More explicitly, the CB-VF $\cbvf(x,T)$ is equal to the initial immediate safety $\terminalpayoff(x)$ when there is a control signal $\usig$ which can prevent the immediate safety along the trajectory $\terminalpayoff\lf(\traj(t)\rg)$ from decaying too quickly, where $\al$ determines the allowed rate of decay (Fig. \ref{fig:cb-vfs}).
\begin{figure}
    \centering
    \includegraphics[width=\columnwidth]{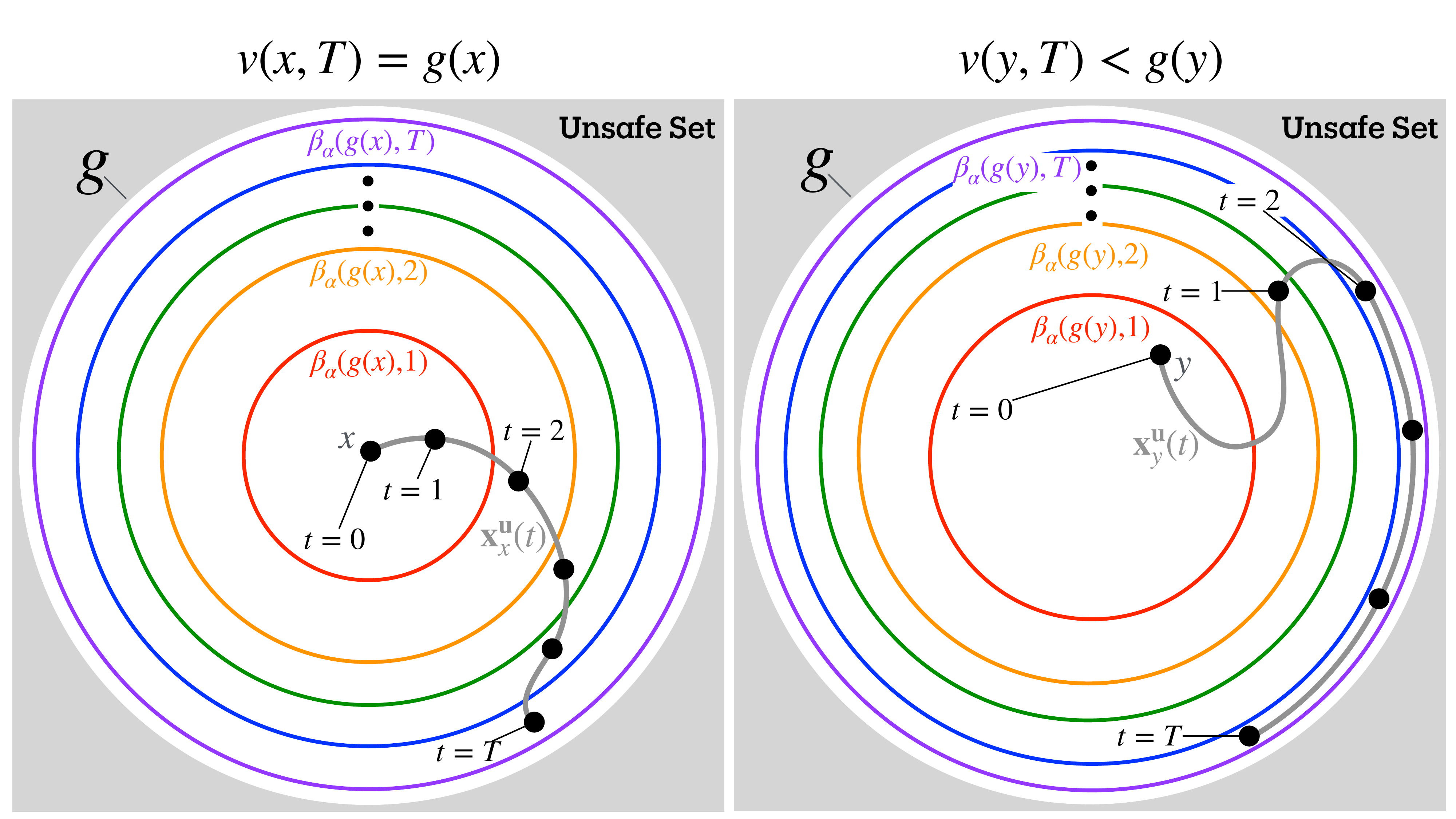}
    \caption{Cartoon showing the level sets of an example immediate safety function $\terminalpayoff:\R^2 \to \R$.
    Given a locally Lipschitz class $\K$ function $\alpha$, a time horizon $T$, and an initial state $x$ for which $\terminalpayoff(x) > 0$, the value $\cbvf(x,T)$ of the corresponding CB-VF is equal to the initial immediate safety $\terminalpayoff(x)$ when the control can keep the immediate safety value $\terminalpayoff\lf(\traj(t)\rg)$ along the trajectory above or arbitrarily close to $\bt_\al(\terminalpayoff(x),t)$ at all times $t \in [0,T]$.
    Otherwise, the CB-VF value is strictly less than the initial immediate safety.}
    \label{fig:cb-vfs}
\end{figure}

\section{Main Theoretical Results}

The first theorem characterizes a CB-VF as the viscosity solution of a certain HJ-PDE.
This is essentially a modification of Theorem 3 in \cite{Choi-Herbert-CDC-CB-VFs-2021} for nonlinear $\al$.

\begin{theorem}\label{theorem:cb-vf}
    Let $\terminalpayoff:\Rn \to \Rge$ be continuous and let $\al$ be locally Lipschitz class $\K$.
    Define the Hamiltonian $\hamiltonian: \Rn \tms \Rge \tms \Rn \to \R$ by \eqref{eqn:hamiltonian}.
    The CB-VF $\cbvf:\Rn \tms \Rge \to \Rge$ corresponding to $\terminalpayoff$ and $\al$ is continuous and is the unique viscosity solution of the HJ-PDE
    \begin{equation}\label{eqn:time-dependent-HJ-PDE}
        \max\lf\{\pp{T} \cbvf - \hamiltonian(x,\cbvf, \nabla \cbvf), \cbvf - \terminalpayoff(x) \rg\} = 0.
    \end{equation}
    for which the initial condition $\cbvf(\cdot,0) \equiv \vcbf(\cdot)$ holds.
\end{theorem}
\begin{proof}
    See Appendix \ref{subsection:appendix-cb-vf-proof}.
\end{proof}
\begin{remark}
    It is standard in HJR to let $t \le 0$ be the initial time, let $0$ be the final time, and derive an analogue of \eqref{eqn:time-dependent-HJ-PDE} via dynamic programming in $t$.
    In this work, it is more clear and convenient to let $0$ be the initial time, let $T \ge 0$ be the final time, and do dynamic programming in $T$.
    It is for this reason that we consider an initial-value problem rather than a terminal-value problem.
    We refer the reader to Chapters 10.1 and 10.3 of \cite{Evans-PDEs} for more on interchanging between such problems.
\end{remark}

The following result establishes an equivalence between non-negative viscosity CBFs, time-invariant CB-VFs, and non-negative functions which provide the Barrier Guarantee.
In particular, all these functions are viscosity solutions of a certain time-independent HJ-PDE.
\begin{theorem}\label{theorem:main-theorem}
    Let $\al$ be locally Lipschitz class $\K$, and let $\vcbf:\R^n \to \Rge$ be continuous.
    The following are equivalent:
    \begin{enumerate}
        \item $\vcbf$ is a viscosity CBF with respect to $\al$.
        \item $\vcbf$ is a viscosity solution of the HJ-PDE
        \begin{equation}\label{eqn:time-independent-HJ-PDE}
            -\min\{\hamiltonian(x,\vcbf,\nabla \vcbf),0\} = 0
        \end{equation}
        where the Hamiltonian $\hamiltonian:\Rn \tms \Rge \tms \Rn \to \R$ is given by \eqref{eqn:hamiltonian}.
        \item The CB-VF $\cbvf$ corresponding to $\vcbf$ and $\al$ is time-invariant, i.e. $\cbvf(\cdot,T) \equiv \vcbf(\cdot)$ for each $T \ge 0$.
        \item $\vcbf$ provides the Barrier Guarantee with respect to $\al$. 
    \end{enumerate}
\end{theorem}
\begin{proof}
    See Appendix \ref{subsection:appendix-main-theorem-proof}.
\end{proof}
Let us make explicit the meaning of the equivalence between 1 and 3.
Recall that the CB-VF $\cbvf$ is always defined with respect to a chosen immediate safety function (typically called $\terminalpayoff$, but in the above theorem it is $\vcbf$).
The implication $1 \to 3$ indicates that if we choose the immediate safety function $h$ to be a viscosity CBF, the corresponding CB-VF will be time-invariant, i.e. always equal to $h$.
The implication $3 \to 1$ indicates that if we happen to choose an immediate safety function $\vcbf$ for which the corresponding CB-VF is time-invariant, then $\vcbf$ must be a viscosity CBF.

For completeness, we provide the following corollary, which does not assume that $\vcbf$ is non-negative.
\begin{corollary}
    Let $\al$ be locally Lipschitz class $\K$, and let $\vcbf:\Rn \to \R$ be continuous.
    Then $\vcbf$ is a viscosity CBF w.r.t. $\al$ iff it provides the Barrier Guarantee w.r.t. $\al$.
\end{corollary}
\begin{proof}
    Immediate from Theorem \ref{theorem:main-theorem} and Lemmas \ref{lemma:non-negative-barrier-guarantee} and \ref{lemma:non-negative-viscosity-cbf}.
\end{proof}
Note that in the above corollary, no mention of value functions were made, but arriving at this conclusion required us to leverage the relationship between CBFs and HJR.
In this sense, the above result is an example of the utility of connecting these two theoretical frameworks.

\begin{corollary}
    Let $\vcbf:\Rn \to \Rge$ be continuous.
    The avoid value function $V:\Rn \tms \Rge \to \Rge$ corresponding to $\vcbf$ is time-invariant iff $\vcbf$ is a viscosity CBF with respect to all locally Lipschitz class $\K$ functions $\al$.
\end{corollary}
\begin{proof}
    It can be checked similarly to steps ($3 \to 4$) and ($4 \to 3$) in Appendix \ref{subsection:appendix-main-theorem-proof} that $V(\cdot,T) \equiv \vcbf(\cdot)$ for all $T \ge 0$ iff $V$ provides the Barrier Guarantee with respect to all locally Lipschitz class $\K$ functions $\al$.
    Theorem \ref{theorem:main-theorem} then gives the desired result.
\end{proof}

\section{Synthesis properties of viscosity CBFs}

Now that we have established the theoretical utility of viscosity CBFs in Theorem \ref{theorem:main-theorem}, let us investigate some of their properties.
Lemma \ref{lemma:viscosity-cbf-is-hj-pde-sub-solution} established an equivalence between viscosity CBFs and viscosity sub-solutions of a certain HJ-PDE.
This equivalence allows us to apply the remarkable array of theory developed for viscosity solutions of HJ-PDEs directly to viscosity CBFs.
As an example of how one can leverage this body of work, we briefly demonstrate ways to synthesize new viscosity CBFs from known viscosity CBFs (or traditional CBFs) via this approach.

\begin{theorem}
    Let $\vcbf_1, \vcbf_2: \Rn \to \R$ be viscosity CBFs with respect to a class $\K$ function $\al$.
    Then $\max\{\vcbf_1, \vcbf_2\}$ is also a viscosity CBF with respect to $\al$.
\end{theorem}
\begin{proof}
    This result follows from Lemma \ref{lemma:viscosity-cbf-is-hj-pde-sub-solution} and Proposition II.2.1(a) in \cite{Bardi-Dolcetta-Optimal-Control}.
\end{proof}

\begin{theorem}
    Let $\vcbf_1, \vcbf_2, \dots: \Rn \to \R$ be a sequence of viscosity CBFs with respect to a class $\K$ function $\al$.
    Suppose $\vcbf_i \to \vcbf$ locally uniformly, where $\vcbf:\Rn \to \Rge$.
    Then $\vcbf$ is also a viscosity CBF with respect to $\al$.
\end{theorem}
\begin{proof}
    Note that $\vcbf$ is automatically continuous by the Uniform Limit Theorem.
    The result then follows from Lemma \ref{lemma:viscosity-cbf-is-hj-pde-sub-solution} and Proposition II.2.2 in \cite{Bardi-Dolcetta-Optimal-Control}.
\end{proof}
\begin{remark}
    The above result also holds if the dynamics themselves change, i.e. $\vcbf_i$ is a viscosity CBF for the dynamics $\dot{\xsig} = f_i(\xsig,\usig)$, so long as the $f_i$ also converge locally uniformly to $f$.
    See Proposition II.2.2 in \cite{Bardi-Dolcetta-Optimal-Control} for details.
\end{remark}

\begin{remark}
    The theorems in this section assume we have already been given viscosity (or traditional) CBFs when synthesizing new ones.
    One may wonder if there is a computational technique to synthesize a viscosity CBF from scratch.
    Indeed, \cite{Choi-Herbert-CDC-CB-VFs-2021} reasons that if a CB-VF $\cbvf$ corresponding to some immediate safety function $\terminalpayoff$ converges as $T \to \ii$, i.e. $\cbvf(\cdot,T) \to \vcbf(\cdot)$, and the limit $\vcbf$ happens to be continuously differentiable, then one can expect $\vcbf$ to be a valid CBF for the largest control-invariant subset within the zero super-level set of $\terminalpayoff$.
    This work suggests that even when not differentiable, one can expect $\vcbf$ will be a viscosity CBF for this subset.
    Future work is required to make this approach precise.
\end{remark}

\section{Conclusion}
We previously noted that a viscosity solution of a partial differential equation is useful because it does not require continuous differentiability (can be too restrictive for existence) but requires more than differentiability almost everywhere (can be too loose for uniqueness).
In the same vein, a viscosity CBF is useful because it does not require continuous differentiability, which can be too restrictive for existence (e.g. no CBF can certify control invariance of the open unit square), but requires more than the CBF inequality holding almost everywhere (which is too loose for safety guarantees).
While this symmetry may at first seem like a coincidence, it is instead a consequence of the fact that viscosity CBFs are indeed HJR value functions (namely time-invariant CB-VFs).

This realization opens up a new approach to CBF theory by characterizing a CBF (in this more generalized sense) as the viscosity sub-solution of an invariant HJ-PDE, which immediately allows us to use the theory developed for these PDEs to obtain theoretical guarantees on viscosity CBFs.
With this connection in mind, the authors expect that advancements in the theory of HJR will be useful for extending the theory of CBFs and vice-versa.

\appendix

\subsection{Proof of Theorem \ref{theorem:cb-vf}}\label{subsection:appendix-cb-vf-proof}
Throughout this part of the appendix, we fix an arbitrary locally Lipschitz class $\K$ function $\al:\Rge \to \Rge$ and an arbitrary continuous function $\terminalpayoff:\Rn \to \Rge$.
We define $\hamiltonian:\Rn \tms \Rge \tms \Rn \to \R$ by \eqref{eqn:hamiltonian} and define $\vanillahamiltonian:\Rn \tms \Rn \to \R$ by $\vanillahamiltonian(x,\lm) = \max_{u \in \uvals} \lm \cdot f(x,u)$.

Additionally, for each $r \ge 0$, we let $[0,b_\al(r))$ be the maximal interval of existence of the initial value problem
\begin{equation}\label{proof:inverse-dynamics}
    \dot{y} = \al(y), \quad y(0) = r,
\end{equation}
and set $A_\al = \{(r,t) \in \Rge \tms \Rge \mid t \in [0,b_\al(r))\}$.
We define $\kp_\al:A_\al \to \Rge$ by for each $r \ge 0$ setting $\kp_\al(r,\cdot):[0,b_\al(r)) \to \Rge$ to be the solution to \eqref{proof:inverse-dynamics}.

\begin{lemma}\label{lemma:derivatives-of-beta}
    For each $r,t > 0$, we have $\pp{t} \bt_\al(r,t) = -\al(\bt_\al(r,t))$ and $\pp{r} \bt_\al(r,t) = \al(\bt_\al(r,t)) / \al(r)$.
\end{lemma}
\begin{proof}
    It follows immediately from the definition of $\bt_\al$ that $\pp{t} \bt_\al(r,t) = -\al(\bt_\al(r,t))$ for all $r > 0$ and $t \ge 0$ (where for $t = 0$ the partial derivative is understood to be one-sided).

    Now fix some particular $r,t > 0$.
    For each $\ta > 0$,
    \begin{equation*}
        \frac{\bt_\al\lf(\bt_\al(r,\ta), t \rg) - \bt_\al(r,t)}{\bt_\al(r,\ta) - r} 
        = \frac{\lf(\bt_\al(r, t + \ta) - \bt_\al(r,t)\rg) / \ta}{(\bt_\al(r,\ta) - r) / \ta}.
    \end{equation*}
    Thus
    \begin{align*}
        \frac{\bt_\al\lf(\bt_\al(r,\ta), t \rg) - \bt_\al(r,t)}{\bt_\al(r,\ta) - r} \overset{\ta \to 0^+}{\longrightarrow} \frac{\pp{t}\bt_\al(r,t)}{\pp{t}\bt_\al(r,0)} = \frac{\al(\bt_\al(r,t))}{\al(r)}.
    \end{align*}

    Next, notice that for each $\ta \in (0,\min\{b_{\al}(r),t\})$,
    \begin{equation*}
        \frac{\bt_\al\lf(\kp_\al(r,\ta), t \rg) - \bt_\al(r,t)}{\kp_\al(r,\ta) - r} 
        = \frac{\lf(\bt_\al(r, t - \ta) - \bt_\al(r,t)\rg) / \ta}{(\kp_\al(r,\ta) - r) / \ta}.
    \end{equation*}
    Thus
    \begin{align*}
        \frac{\bt_\al\lf(\kp_\al(r,\ta), t \rg) - \bt_\al(r,t)}{\kp_\al(r,\ta) - r} \overset{\ta \to 0^+}{\longrightarrow} \frac{-\pp{t}\bt_\al(r,t)}{\pp{t} \kp_\al(r,0)} = \frac{\al(\bt_\al(r,t))}{\al(r)}.
    \end{align*}
    The desired result follows from noticing these these limits are in fact the left and right derivatives of $\bt_\al(r,t)$ w.r.t. $r$.
\end{proof}

\begin{lemma}\label{lemma:derivatives-of-kappa}
    Let $(t,r) \in A_\al$ with $t, r > 0$.
    Then $\pp{t} \kp_\al(r,t) = \al(\kp_\al(r,t))$ and $\pp{r} \kp_\al(r,t) = \al(\kp_\al(r,t)) / \al(r)$.
\end{lemma}
\begin{proof}
The result for the partial derivative w.r.t. $t$ follows directly from the definition of $\kp_\al$.
Also by definition of $\kp_\al$, for all $(\rh,\ta) \in A_\al$ we have $\bt_\al(\kp_\al(\rh,\ta),\ta) = \rh$.
Then by the inverse function theorem and Lemma \ref{lemma:derivatives-of-beta},
\begin{equation*}
\pp{r} \kp_\al(r,t) = \frac{\al(\kp_\al(r,t))}{\al(\bt_\al(\kp_\al(r,t),t))} = \frac{\al(\kp_\al(r,t))}{\al(r)}.
\end{equation*}
\end{proof}

\begin{lemma}\label{lemma:beta-comparison-principle}
Let $R > 0$.
There is an $L > 0$ such that $\bt_\al(r,t) \ge r e^{-Lt}$ for all $r \in [0,R]$ and $t \ge 0$.
\end{lemma}
\begin{proof}
    Because $\al$ is locally Lipschitz, we can choose an $L > 0$ such that $\al(r) \le Lr$ for each $r \in [0,R]$.
    Fix a particular $r \in [0,R]$.
    Recall that $\bt_\al(r,\cdot)$ is the solution to the ODE $\dot{y} = -\al(y)$ corresponding to the initial condition $y(0) = r$.
    For all $t \ge 0$, we have $\bt_\al(r,t) \le R$ so that
    \begin{equation*}
        \pp{t}\bt_\al(r,t) = - \al(\bt_\al(r,t)) \\
        \ge -L\bt_\al(r,t).
    \end{equation*}
    The result then follows from the comparison principle.
\end{proof}

\begin{lemma}\label{lemma:v-w-equivalence}
    Suppose that $\cbvf,w:\Rn \tms \Rge \to \Rge$ are continuous and satisfy $\bt_\al(\cbvf(x,T),T) = w(x,T)$ for all $(x,T) \in \Rn \tms \Rge$.
    Then $w$ is a viscosity solution of
    \begin{equation}\label{eqn:transformed-time-dependent-HJ-PDE}
        \max\lf\{ \pp{T} w - \vanillahamiltonian(x,\nabla w), w - \bt_{\al}(\terminalpayoff(x),T) \rg\} = 0
    \end{equation}
    iff $\cbvf$ is viscosity solution of \eqref{eqn:time-dependent-HJ-PDE}.
\end{lemma}
\begin{proof}
    ($\implies$) Suppose that $w$ is a viscosity solution of \eqref{eqn:transformed-time-dependent-HJ-PDE}.
    We first establish that $\cbvf$ is a viscosity sub-solution of $\eqref{eqn:time-dependent-HJ-PDE}$.
    Let $(x_0,T_0) \in \Rn \tms (0,\ii)$. 
    Suppose that $\ph:\Rn \tms (0,\ii) \to \R$ is continuously differentiable and $\cbvf - \ph$ has a local maximum at $(x_0,T_0)$.
    Without loss of generality we may assume $\ph(x_0,T_0) = \cbvf(x_0,T_0)$.
    We wish to show that
    \begin{align}\label{proof:change-of-coordinates-wwts-1}
        \max\Big\{ &\pp{T} \ph(x_0,T_0) - \hamiltonian(x_0,\ph(x_0,T_0), \nabla \ph(x_0,T_0)),\nonumber\\
        &\ph(x_0,T_0) - \terminalpayoff(x_0) \Big\} \le 0.
    \end{align}

    First suppose $\cbvf(x_0,T_0) = \ph(x_0,T_0) = 0$.
    Since $\cbvf - \ph$ has a local maximum at $(x_0,T_0)$ and since $\cbvf$ is non-negative it follows that $\ph$ has a local minimum at $(x_0,T_0)$.
    Thus $\pp{T} \ph(x_0,T_0) = 0$ and $\nabla \ph(x_0,T_0) = 0$, so \eqref{proof:change-of-coordinates-wwts-1} holds.

    Next suppose that $\cbvf(x_0,T_0) = \ph(x_0,T_0) > 0$.
    Choose an open neighborhood $\Nc \sbs \Rn \tms (0,\ii)$ of $(x_0,T_0)$ such that $\ph > 0$ everywhere in $\Nc$.
    Let $\ps: \Nc \to \Rge$ be defined by $\ps(x,T) = \bt_\al(\ph(x,T),T)$.
    Then $\ps(x_0,T_0) = \bt_\al(\cbvf(x_0,T_0),T_0) = w(x_0,T_0)$ and $w|\Nc - \ps$ has a local maximum at $(x_0,T_0)$.
    Since $w|\Nc$ is a viscosity solution of $\eqref{eqn:transformed-time-dependent-HJ-PDE}$ on $\Nc$, it follows that
    \begin{align*}
        \max\Big\{ &\pp{T} \ps(x_0,T_0) - \vanillahamiltonian(x_0,\nabla \ps(x_0,T_0)),\\
        &\ps(x_0,T_0) - \bt_{\al}(\terminalpayoff(x_0),T_0) \Big\} \le 0.
    \end{align*}
    For convenience, let $a = \frac{\partial}{\partial t}\bt_\al\lf( \ph(x_0,T_0),T_0 \rg)$ and let $b = \frac{\partial}{\partial r}\bt_\al\lf( \ph(x_0,T_0),T_0 \rg)$ so that
    \begin{align*}
        \max\Big\{& a + b \pp{T} \ph(x_0,T_0) - \vanillahamiltonian\lf(x_0, b \nabla \ph(x_0,T_0) \rg),\\
        &\bt_\al(\ph(x_0,T_0),T_0) - \bt_\al(\terminalpayoff(x_0),T_0) \Big\} \le 0.
    \end{align*}
    Since $a/b = -\al(\ph(x_0,T_0))$ with $b > 0$ and since $\bt_\al(\cdot,T_0)$ is strictly increasing, the above inequality implies \eqref{proof:change-of-coordinates-wwts-1}. 

    We next establish that $\cbvf$ is a viscosity super-solution of $\eqref{proof:change-of-coordinates-wwts-1}$.
    Let $(x_0,T_0) \in \Rn$. 
    Suppose that $\ph:\Rn \tms (0,\ii) \to \R$ is continuously differentiable and $\cbvf - \ph$ has a local minimum at $(x_0,T_0)$.
    Without loss of generality we may assume $\ph(x_0,T_0) = \cbvf(x_0,T_0)$.
    We wish to show that
    \begin{align}\label{proof:change-of-coordinates-wwts-2}
        \max\Big\{ &\pp{T} \ph(x_0,T_0) - \hamiltonian(x_0,\ph(x_0,T_0), \nabla \ph(x_0,T_0)),\nonumber\\
        &\ph(x_0,T_0) - \terminalpayoff(x_0) \Big\} \ge 0.
    \end{align}
    
    First suppose $\cbvf(x_0,T_0) = \ph(x_0,T_0) = 0$.
    Since $\cbvf - \ph$ has a local minimum at $(x_0,T_0)$,
    we can choose some compact neighborhood $\Nc \sbs \Rn \tms (0,\ii)$ of $(x_0,T_0)$ such that $\cbvf \ge \ph$ in $\Nc$.
    Let $R = \max_{(x,T) \in \Nc} \cbvf(x,T) + 1$.
    By Lemma \ref{lemma:beta-comparison-principle}, we can choose an $L > 0$ such that $\bt_\al(r,t) \ge r e^{-Lt}$ for each $r \in [0,R]$ and $t \ge 0$.
    Moreover, since $\Nc$ is compact, we can choose an $S \in \R$ such that $S > T$ for each $(x,T) \in N$.
    Setting $c = e^{-LS}$, it follows that $w(x,T) - c\ph(x,T) = \bt(v(x,T),T) - c\ph(x,T) \ge 0$ for each $(x,T) \in N$.
    Since $w(x_0,T_0) - c\ph(x_0,T_0) = 0$, it follows that $w - c\ph$ has a local minimum at $(x_0,T_0)$.
    Eq. $\eqref{proof:change-of-coordinates-wwts-2}$ then follows from
    \begin{align*}
        \max\Big\{& c \pp{T} \ph(x_0,T_0) - \vanillahamiltonian\lf(x_0, c \nabla \ph(x_0,T_0) \rg),\\
        & - \bt_\al(\terminalpayoff(x_0),T_0) \Big\} \ge 0.
    \end{align*}

    The case $\cbvf(x_0,T_0) = \ph(x_0,T_0) > 0$ follows analogously to the same case in the sub-solution proof.

    ($\impliedby$) Suppose that $\cbvf$ is a viscosity solution of \eqref{eqn:time-dependent-HJ-PDE}.
    We first establish that $w$ is a viscosity sub-solution of \eqref{eqn:transformed-time-dependent-HJ-PDE}.
    Let $(x,T) \in \Rn \tms (0,\ii)$. 
    Suppose that $\ph:\Rn \tms (0,\ii) \to \R$ is continuously differentiable and $w - \ph$ has a local maximum at $(x_0,T_0)$.
    Without loss of generality we may assume $\ph(x_0,T_0) = \cbvf(x_0,T_0)$.
    We wish to show that
    \begin{align}\label{proof:change-of-coordinates-wwts-3}
        \max\Big\{ &\pp{T} \ph(x_0,T_0) - \vanillahamiltonian(x_0, \nabla \ph(x_0,T_0)),\nonumber\\
        &\ph(x_0,T_0) - \bt_\al(\terminalpayoff(x_0),T_0) \Big\} \le 0.
    \end{align}
    
    First suppose $w(x_0,T_0) = \ph(x_0,T_0) = 0$.
    Since $w - \ph$ has a local maximum at $(x_0,T_0)$ and since $w$ is non-negative it follows that $\ph$ has a local minimum at $(x_0,T_0)$.
    Thus $\pp{T} \ph(x_0,T_0) = 0$ and $\nabla \ph(x_0,T_0) = 0$, so \eqref{proof:change-of-coordinates-wwts-3} holds.

    Next suppose that $w(x_0,T_0) = \ph(x_0,T_0) > 0$.
    Choose a neighborhood $\Nc \sbs \Rn \tms (0,\ii)$ of $(x_0,T_0)$ such that $(\ph(x,T),T) \in A_\al$ and $\ph(x,T) > 0$ for each $(x,T) \in \Nc$.
    Let $\ps: \Nc \to \Rge$ be defined by $\ps(x,T) = \kp_\al(\ph(x,T),T)$.
    Then $\ps(x_0,T_0) = \kp_\al(w(x_0,T_0),T_0) = \cbvf(x_0,T_0)$ and $\cbvf|\Nc - \ps$ has a local maximum at $(x_0,T_0)$.
    Since $\cbvf|\Nc$ is a viscosity solution of $\eqref{eqn:time-dependent-HJ-PDE}$ on $\Nc$, it follows that
    \begin{align*}
        \max\Big\{ &\pp{T} \ps(x_0,T_0) - \hamiltonian(x_0,\ps(x_0,T_0), \nabla \ps(x_0,T_0)),\nonumber\\
        &\ps(x_0,T_0) - \terminalpayoff(x_0) \Big\} \le 0.
    \end{align*}
    For convenience, let $a = \frac{\partial}{\partial t}\kp_\al\lf( \ph(x_0,T_0),T_0 \rg)$ and let $b = \frac{\partial}{\partial r}\kp_\al\lf( \ph(x_0,T_0),T_0 \rg)$ so that
    \begin{align*}
        \max\Big\{& a + b \pp{T} \ph(x_0,T_0) \\
        &- \hamiltonian\lf(x_0, \kp_\al(\ph(x_0,T_0),T_0), b \nabla \ph(x_0,T_0) \rg),\\
        &\kp_\al(\ph(x_0,T_0),T_0) - \terminalpayoff(x_0) \Big\} \le 0.
    \end{align*}
    Since $a = \al(\kp_\al(\ph(x_0,T_0),T_0)$ and $b > 0$ by Lemma \ref{lemma:derivatives-of-kappa}, the above inequality implies \eqref{proof:change-of-coordinates-wwts-3}.

    We next establish that $w$ is a viscosity super-solution of \eqref{eqn:transformed-time-dependent-HJ-PDE}.
    Let $(x_0,T_0) \in \Rn \tms (0,\ii)$. 
    Suppose that $\ph:\Rn \tms (0,\ii) \to \R$ is continuously differentiable and $w - \ph$ has a local minimum at $(x_0,T_0)$.
    Without loss of generality we may assume $\ph(x_0,T_0) = \cbvf(x_0,T_0)$.
    We wish to show that
    \begin{align}\label{proof:change-of-coordinates-wwts-4}
        \max\Big\{ &\pp{T} \ph(x_0,T_0) - \vanillahamiltonian(x_0, \nabla \ph(x_0,T_0)),\nonumber\\
        &\ph(x_0,T_0) - \bt_\al(\terminalpayoff(x_0),T_0) \Big\} \ge 0.
    \end{align}
    
    First suppose $w(x_0,T_0) = \ph(x_0,T_0) = 0$.
    Since $\cbvf(x_0,T_0) = 0$ and $\cbvf \ge w$ and since $w - \ph$ has a local minimum at $(x_0,T_0)$, it follows that $\cbvf - \ph$ has a local minimum at $(x_0,T_0)$.
    Thus
    \begin{align*}
        \max\Big\{ &\pp{T} \ph(x_0,T_0) - \hamiltonian(x_0, 0, \nabla \ph(x_0,T_0)), - \terminalpayoff(x_0) \Big\} \ge 0.
    \end{align*}
    But the above implies \eqref{proof:change-of-coordinates-wwts-4}.
    The case $w(x_0,T_0) = \ph(x_0,T_0) > 0$ is similar to the sub-solution argument.

\end{proof}

\begin{lemma}\label{lemma:existence-and-uniqueness}
    Let $\terminalpayoff:\Rn \to \Rge$ be continuous, and let $\al:\Rge \to \Rge$ be locally Lipschitz class $\K$.
    Define $w:\Rn \tms \Rge \to \Rge$ by 
    $$w(x,T) = \sup_{\usig \in \usigs} \min_{t \in [0,T]} \bt_\al \lf(\terminalpayoff \lf( \traj(t) \rg), T - t\rg).$$
    Then $w$ is continuous and is the unique viscosity solution of \eqref{eqn:transformed-time-dependent-HJ-PDE}.
\end{lemma}
\begin{proof}
Continuity of $w$ follows from observing
\begin{align*}
    |w(x,T) - w(y,S)| \le \sup_{\usig \in \usigs} &\Big\vert \min_{t \in [0,T]}  \bt_\al \lf( \terminalpayoff \lf( \traj(t) \rg), T - t \rg) \\
    &- \min_{s \in [0,S]} \bt_\al \lf(\terminalpayoff \lf( \xsig_{y}^\usig(s) \rg), S - t\rg)\Big\vert
\end{align*}
and recalling that $\bt_\al$ is continuous and the dynamics $f$ are Lipschitz.
That $w$ is a viscosity solution of $\eqref{eqn:transformed-time-dependent-HJ-PDE}$ is a standard dynamic programming argument, see e.g. \cite{Fisac-Sastry-Reach-Avoid-2015}.
The uniqueness result is a straightforward extension of Theorems III.3.12 and III.3.15 in \cite{Bardi-Dolcetta-Optimal-Control}. (See \cite{Barron_1989} for another approach under slightly different assumptions.)
\end{proof}

\begin{proof}[Proof of Theorem \ref{theorem:cb-vf}]
    Let $w$ be as in Lemma \ref{lemma:existence-and-uniqueness}.
    For each $(x,T) \in \Rn \tms (0,\ii)$, we have $\cbvf(x,T) = \kp_\al(w(x,T),T)$.
    As $w$ and $\kp_\al$ are continuous, $\cbvf$ is also.
    The result follows from Lemmas \ref{lemma:v-w-equivalence} and \ref{lemma:existence-and-uniqueness}.
\end{proof}

\subsection{Proof of Theorem \ref{theorem:main-theorem}}\label{subsection:appendix-main-theorem-proof}
\begin{proof}[Proof of Theorem \ref{theorem:main-theorem}]
($1\to2$) Suppose $\vcbf$ is a viscosity CBF w.r.t. $\al$.
Let $\ph:\Rn \to \R$ be continuously differentiable, and suppose $\vcbf - \ph$ has a local maximum at $x_0 \in \Rn$.
Then $\hamiltonian(x_0,\vcbf(x_0),\nabla \ph(x_0)) \ge 0$ by Lemma \ref{lemma:viscosity-cbf-is-hj-pde-sub-solution}.
Thus $-\min\{\hamiltonian(x_0,\vcbf(x_0),\nabla \ph(x_0)),0\} \le 0$.
On the other hand, for any continuously differentiable $\ph:\Rn \to \R$ and any $x_0 \in \Rn$, we trivially have $-\min\{\hamiltonian(x_0,\vcbf(x_0),\nabla \ph(x_0)),0\} \ge 0$.

($2\to1$) Suppose $\vcbf$ is a viscosity solution of \eqref{eqn:time-independent-HJ-PDE}.
Let $\ph:\Rn \to \R$ be continuously differentiable, and suppose $\vcbf - \ph$ has a local maximum at $x_0 \in \Rn$.
Then $-\min\{\hamiltonian(x_0,\vcbf(x_0),\nabla \ph(x_0)),0\} \le 0$, so that $\hamiltonian(x_0,\vcbf(x_0),\nabla \ph(x_0)) \ge 0$.

($2\to3$) Suppose $\vcbf$ is a viscosity solution of $\eqref{eqn:cbvf}$.
Define $\tilde{\cbvf}:\Rn \tms \Rge \to \R$ by
$\tilde{\cbvf}(x,T) = \vcbf(x)$.
It follows that $\vcbf$ is also viscosity solution of \eqref{eqn:time-dependent-HJ-PDE}.
Then $\cbvf = \tilde{\cbvf}$ by Theorem \ref{theorem:cb-vf}.

($3\to2$) Immediate from Theorem \ref{theorem:cb-vf}.

($3\to4$) Suppose that for each $T \ge 0$, we have $\cbvf(\cdot,T) \equiv \vcbf(\cdot)$.
Then for each $x \in \Rn$, we in particular have
\begin{align*}
    \bt_\al(\vcbf(x),1) &= \sup_{\usig \in \usigs} \min_{t \in [0,1]} \bt_\al\lf(\vcbf(\traj(t)), 1 - t \rg).
\end{align*}

Now, fix some $x \in \Rn$ and $\te \in [0,1)$.
We construct a control signal $\usig^* \in \usigs$ recursively.
Set $x_0 = x$ and $\te_0 = \te$.
Having defined $x_n \in \Rn$ and $\te_n \in [0,1)$, choose $\usig_{n+1} \in \usigs$ and $\te_{n+1} \in [0,1)$ such that
$\bt_\al(\te_n\vcbf(x_n),1) \le \min_{t \in [0,1]}  \bt_\al(\te_{n+1} \vcbf(\xsig_{x_{n}}^{\usig_{n+1}}(t)), 1-t)$.
It follows that for each $t \in [0,1]$, $\bt_\al(\bt_\al(\te_n \vcbf(x_n), t),1-t) = \bt_\al(\te_n \vcbf(x_n), 1) \le \bt_\al(\te_{n+1}\vcbf(\xsig_{x_n}^{\usig_{n+1}}(t)) , 1-t)$.
Thus 
\begin{equation}\label{proof:equivalence-theorem-3-4-1}
    \bt_\al(\te_n \vcbf(x_n), t) \le \te_{n+1}\vcbf(\xsig_{x_n}^{\usig_{n+1}}(t)) \textall t \in [0,1].
\end{equation}
Now set $x_{n+1} = \xsig_{x_{n}}^{\usig_{n+1}}(1)$.
We define $\usig \in \usigs$ by setting
$\usig(\cdot) = \usig_i(\cdot)$ on $[i-1,i)$ for each $i \in \N$.
Thus for each $n = 0,1,\dots$ and $t \in [0,1]$ we have
\begin{equation}\label{proof:equivalence-theorem-3-4-2}
    \startraj(n+t) = \xsig_{x_{n}}^{\usig_{n+1}}(t) \textall t \in [0,1].
\end{equation}

We show by induction that for each $k = 0,1,\dots$ we also have $\bt_\al(\te \vcbf(x), k + t) \le \te_{k+1} \vcbf(\xsig_{x}^{\usig^*}(k + t))$ for all $t \in [0,1]$, from which the result follows.
Indeed, the base case is immediate from \eqref{proof:equivalence-theorem-3-4-1}-\eqref{proof:equivalence-theorem-3-4-2}.
Now, suppose the induction hypothesis holds for $k$.
Then by \eqref{proof:equivalence-theorem-3-4-2}, $\bt_\al(\te \vcbf(x), k + 1) \le \te_{k+1} \vcbf(x_{k+1})$ so that for all $t \in [0,1]$,
\begin{align*}
    \bt_\al&(\te \vcbf(x), k + 1 + t) =
    \bt_\al(\bt_\al(\te \vcbf(x), k + 1),t) \\
    &\le \bt_\al(\te_{k+1} \vcbf(x_{k+1}), t) \le \te_{k+2}\vcbf(\xsig_{x_{k+1}}^{\usig_{k+2}}(t))\\
    &= \te_{k+2}\vcbf(\startraj(k+1+t)).
\end{align*}

($4\to3$)
Suppose $g$ provides the Barrier Guarantee w.r.t. $\al$.
Fix some particular $\te \in [0,1)$ and $x \in \Rn$, and choose a $\usig \in \usigs$ such that
$\bt_\al(\te\vcbf(x), t) \le \vcbf\lf(\traj(t)\rg)$
for all $t \ge 0$.
Then $\bt_\al(\te\vcbf(x), T) = \bt_\al\lf(\bt_\al(\te\vcbf(x), t), T-t\rg) \le \bt_\al\lf(\vcbf\lf(\traj(t)\rg), T-t \rg)$ for all $t \ge 0$ and $T \ge t$, .
Thus $\bt_\al(\te\vcbf(x), T) \le \sup_{\tilde{\usig} \in \usigs} \min_{t \in [0,T]} \bt_\al\lf(\vcbf\lf(\tildetraj(t)\rg), T-t \rg) \le \bt_\al(\vcbf(x), T)$ for all $T \ge 0$.
Letting $\te \to 1$ shows that $\bt_\al(\vcbf(x), T) = \sup_{\tilde{\usig} \in \usigs} \min_{t \in [0,T]} \bt_\al\lf(\vcbf\lf(\tildetraj(t)\rg), T-t \rg)$.
But then $\cbvf(\cdot,T) \equiv \vcbf(\cdot)$ for each $T \ge 0$.
\end{proof}

\subsection{Control Invariance Counter-Example}\label{subsection:appendix-counter-example}
It is not true that existence of a CBF $h$ generally certifies control-invariance of its \textit{non-strict} zero super-level set $\safeset_{\ge 0} := \{x \in \Rn \mid h(x) \ge 0\}$.
Consider the system in $\R^2$ with dynamics
$\dot{\xsig}_1 = 0, \dot{\xsig}_2 = \usig$ inside the unit ball, and 
$\dot{\xsig}_1 = \xsig_1(1 - \sqrt{\xsig_1^2 + \xsig_2^2}) / \sqrt{\xsig_1^2 + \xsig_2^2}$, $\dot{\xsig}_2 = \xsig_2(1 - \sqrt{\xsig_1^2 + \xsig_2^2}) / \sqrt{\xsig_1^2 + \xsig_2^2} + \usig$ outside of it,
where the admissible control set is $\uvals = \{-1,+1\}$.
Define the function $h:\R^2 \to \R^2$ by $h(x_1,x_2) = 1 - x_1^2 - x_2^2$.
To see that $h$ is indeed a CBF, observe that $\max_{u \in \mathcal{U}} \nabla h(x_1,x_2) \cdot f((x_1,x_2),u) = 2|x_2| \ge -h(x_1, x_2)$ inside the unit ball and
$\max_{u \in \mathcal{U}} \nabla h(x_1,x_2) \cdot f((x_1,x_2),u) = 2 \sqrt{x_1^2 + x_2^2} (\sqrt{x_1^2 + x_2^2} - 1) + 2|x_2| \ge -h(x_1, x_2)$ outside the unit ball.

We will show that the set $\safeset_{\ge 0} := \{(x_1,x_2) \in \R^2 \mid h(x_1,x_2) \ge 0\} = \{(x_1,x_2) \in \R^2 \mid x_1^2 + x_2^2 \le 1\}$ is not control-invariant.
To do so, consider the initial state $(\ba{x}_1, \ba{x}_2) := (1,0) \in \safeset_{\ge 0}$, and let $\usig \in \usigs$ (i.e. $\usig:\Rge \to \uvals$ be an arbitrary measurable control signal that takes on only admissible control values).
For convenience, let $(\xsig_1(\cdot),\xsig_2(\cdot)) \equiv \xsig_{(\ba{x}_1,\bar{x}_2)}^{\usig}(\cdot)$.

Suppose that indeed $\xsig_1(t)^2 + \xsig_2(t)^2 \le 1$ for all $t \ge 0$.
Then $\dot{\xsig}_1(\cdot) \equiv 0$ so that $\xsig_1(\cdot) \equiv 1$ and $\xsig_2(\cdot) \equiv 0$.

It follows that $\int_0^{t} \usig(\ta) \dee \ta = \xsig_2(t) = 0$ for all $t \ge 0$.
Thus $\int_a^b \usig(\ta) \dee \ta = 0$ for all $a,b > 0$.
But the Lebesgue Differentiation Theorem (Theorem 3.21 in \cite{Folland-Real-Analysis}), $\usig(t) = 0$ for a.e. $t \ge 0$, violating the assumption that $\usig(\cdot) \in \{-1,+1\}$.

\section*{Acknowledgment}
The authors thank Dr. Jorge Cort\'{e}s for useful discussions.

\bibliographystyle{ieeetr}
\bibliography{references}
\end{document}